\documentclass[letterpaper, 10 pt, conference]{ieeeconf}  
\IEEEoverridecommandlockouts                              

\usepackage{graphicx} 
\usepackage{amsmath} 
\usepackage{amssymb}  
\usepackage{mathtools}
\usepackage{hyperref}
\usepackage{xcolor}

\usepackage{subcaption}

\usepackage[noend]{algpseudocode}
\usepackage{algorithm}

\usepackage{bbm}

\usepackage{physics}

\newcommand{\qed}{$\square$}
\newcommand{\defeq}{\vcentcolon=}

\newcommand{\dgap}{d}

\DeclareMathOperator*{\argmax}{arg\,max}

\newtheorem{property}{Property}
\newtheorem{thm}{Theorem}
\newtheorem{defn}{Definition}
\newtheorem{prop}{Proposition}

\newtheorem{assum}{Assumption}
\newtheorem{remark}{Remark}

\let\emptyset\varnothing

\usepackage{tikz}
\usetikzlibrary{shapes.geometric, arrows.meta, positioning}
\tikzstyle{startstop} = [rectangle, rounded corners, minimum width=1cm, minimum height=2cm,text centered, draw=black, fill=pink, text width=0.5cm]
\tikzstyle{IO} = [ellipse, minimum width=2.5cm, minimum height=0.1cm,text centered, draw=black, text width=3.5cm]
\tikzstyle{arrow} = [thick,->,>={Stealth}]
\tikzset{font={\fontsize{10pt}{12}\selectfont}}

\usepackage{lipsum}

\newcommand\blfootnote[1]{%
  \begingroup
  \renewcommand\thefootnote{}\footnote{#1}%
  \addtocounter{footnote}{-1}%
  \endgroup
}

\usepackage{tikz}
\usetikzlibrary{patterns}
\usepackage{pgfplots}
\tikzstyle{bag} = [align=left]
\usetikzlibrary{shapes.misc}

\usetikzlibrary{decorations.pathmorphing}
\tikzset{zigzag/.style={decorate,decoration=zigzag}}
\definecolor{sandybrown}{rgb}{0.96, 0.64, 0.38}

\usepackage[normalem]{ulem}
\newif\ifremove
\global\removetrue

\global\removefalse
\ifremove
\renewcommand\sout[1]{}
\fi

\title{\LARGE \bf
Continuous-time Data-driven Barrier Certificate Synthesis}

\author{Luke Rickard$^{1}$, Alessandro Abate$^{2}$ and Kostas Margellos$^{1}$
\thanks{$^{1}$Department of Engineering Science, University of Oxford}
\thanks{$^{2}$Department of Computer Science, University of Oxford}
\thanks{This work was supported by the EPSRC Centre for Doctoral Training in Autonomous Intelligent Machines and Systems EP/S024050/1.}
}

\begin{document}

\maketitle
\thispagestyle{empty}
\pagestyle{empty}

\begin{abstract}

We consider the problem of verifying safety for continuous-time dynamical systems. 
Developing upon recent advancements in data-driven verification, we use only a finite number of sampled trajectories to learn a \emph{barrier certificate}, namely 
a function which verifies safety.
We train a safety-informed neural network to act as this certificate, with an appropriately designed loss function to encompass the safety conditions. 
In addition, we provide probabilistic generalisation guarantees from discrete samples 
of continuous trajectories, to unseen continuous ones.
Numerical results demonstrate the efficacy of our approach and contrast it with related results in the literature.
\hypersetup{hidelinks}\blfootnote{For the purpose of Open Access, the authors
have applied a CC BY public copyright licence to any Author Accepted
Manuscript (AAM) version arising from this submission.}
\end{abstract}

\section{Introduction}

Ensuring the safety of continuous-time dynamical systems is of critical importance in an increasingly autonomous world~\cite{DBLP:conf/hybrid/EdwardsPA24,DBLP:journals/tac/NejatiLJSZ23,DBLP:journals/corr/abs-2502-05510}. 
It is often infeasible to model system behaviour precisely, thus making direct use of system data to verify behaviour is of interest~\cite{DBLP:conf/memocode/Abate17,DBLP:conf/hybrid/KozarevQHT16}.
A technique for verifying properties of dynamical systems involves discretising the state space \cite{DBLP:journals/jair/BadingsRAPPSJ23}, under approximation guarantees,  and verifying the resulting model. 
Alternatively, the use of \emph{certificates}~\cite{DBLP:conf/eucc/AmesCENST19,DBLP:conf/hybrid/PrajnaJ04} allows one to directly analyse the continuous-state system.
These certificates map the system's states to real values, and exhibit certain properties that are relevant for analysis: here, in particular, we construct safety certificates for continuous-time systems, but extensions to the more complex certificates in~\cite{DBLP:journals/corr/abs-2502-05510} are possible.

There are a number of techniques for synthesising such certificates. 
In the case that an exact model is known, one can use a polynomial function as a certificate to formulate a convex sum-of-squares problem~\cite{DBLP:conf/amcc/Papachristodoulou05a}. 
Recent work in this area investigated the use of neural networks (which represent a class of general function approximators) as certificates~\cite{DBLP:conf/hybrid/EdwardsPA24}. 

Obtaining a model of the system, however, is generally a difficult task. 
It requires domain-specific knowledge and, since parts of the system may often not be known exactly, such a model will only constitute an approximation. 
To alleviate these issues, we turn to model-free data-driven techniques.  
One method for employing data in certificate synthesis is through the use of state pairs (i.e. states, and next-states), sampled from across the domain of interest. 
Such techniques are investigated in~\cite{DBLP:journals/tac/NejatiLJSZ23} for deterministic systems, and in~\cite{DBLP:journals/automatica/SalamatiLSZ24} for stochastic systems. 
Both these works make use of the techniques in~\cite{DBLP:journals/jota/KanamoriT12,6832537} to bound the distance between what is referred to as a robust program, and its sample-based counterpart. 
As discussed in~\cite[Remark 3.9]{6832537}, such techniques exhibit an exponential growth in the dimension of the sampling space (here the state space), and also require access to the underlying probability distribution. 

Alternatively, one can consider using entire trajectories as samples, hence only requiring access to runs compatible with the dynamics of the system. 
This is explored in~\cite{DBLP:journals/corr/abs-2502-05510} for discrete-time systems. 
This work leverages the so-called scenario approach~\cite{Scen_approach_book,DBLP:journals/mp/GarattiC22} and the notion of compression \cite{DBLP:journals/tac/MargellosPL15,DBLP:journals/jmlr/CampiG23}, and provides a constructive algorithm for the general pick-to-learn framework~\cite{DBLP:conf/nips/PaccagnanCG23}, to provide a probably approximately correct (PAC) bound on the correctness of certificates for newly sampled discrete trajectories. 
However, these guarantees are no longer valid when it comes to continuous-time systems.
Here, we extend these developments to allow constructing a safety/barrier certificate based on a finite number of discretised trajectories of the continuous-time system; however, we provide guarantees on the safety of a new unseen continuous-time trajectory. 
The latter does not follow from the results in~\cite{DBLP:journals/corr/abs-2502-05510}.

Our contributions can be summarised as follows:
\begin{itemize}
    \item We develop novel theoretical results that extend the results of~\cite{DBLP:journals/corr/abs-2502-05510} (which was limited in scope to discrete-time systems) to verify entire continuous trajectories, using discretised approximations for computations;
    \item We complement the results in \cite{DBLP:journals/tac/NejatiLJSZ23}, following an alternative approach for data-driven certificate synthesis which, however, does not scale exponentially with respect to the dimension of the underlying state space; 
    \item Our approach is entirely data-driven, avoiding the use of satisfiability modulo theories (SMT) solvers, which are computationally expensive and require a system model.
\end{itemize}


\emph{Notation.}
We use $\{\xi_k\}_{k=0}^K$ to denote a sequence indexed by $k \in \{0,1,\dots,K \}$. $B \models \psi$ defines condition satisfaction, i.e., it evaluates to true if the quantity $B$ on the left satisfies the condition $\psi$ on the right.  Using $\not\models$ represents the logical inverse of this (i.e., condition dissatisfaction). By $(\forall \xi \in \Xi) B\models\psi(\xi)$ we mean that some quantity $B$ satisfies a condition $\psi$ which, in turn, depends on some parameter $\xi$, for all $\xi \in \Xi$. 
We use $\xi_{[0,k]}$ to refer to a subsequence $\{\xi_0,\dots,\xi_k\}$ of a sequence.

\section{Safety and Barrier Certificates}
\label{sec:certs}

We focus on barrier certificates as the tool to verify safety; however, our techniques naturally extend to more complex certificates, as in \cite{DBLP:journals/corr/abs-2502-05510}. 
Consider a bounded state space $X \subseteq \mathbb{R}^n$, and a continuous-time dynamical system whose evolution starts at an initial state $x(0) \in X_I$, where $X_I\subseteq X$ denotes the set of possible initial conditions. 
From an initial state, we unravel a finite trajectory, i.e., a continuous sequence of states $\xi = \{x(t)\}_{t\in[0,T]}$, where $T\in \mathbb{R}$ and $x(t) \colon [0,T] \rightarrow \mathbb{R}^n$, by following the dynamics
\begin{equation}
\label{eq:Dyn}
    \dot{x}(t)=f(x(t)).
\end{equation}
We only require $f \colon X \rightarrow \mathbb{R}^n$ to be Lipschitz continuous, thus allowing for existence and uniqueness of solutions. 
The set of all possible trajectories $\Xi \subseteq X_I \times X^{(0,T]}$ is then the set of all trajectories starting from the initial set $X_I$.

In Section~\ref{sec:learn_certs}, we discuss how to use a finite set of trajectories to build safety certificates, and accompany them with generalisation guarantees with respect to their validity for a new, unseen trajectory. 
We will also consider time-discretised versions of continuous-time trajectories. To this end, 
define the time-discretised trajectory
\begin{equation}
    \tilde{\xi}=\{x(t)\}_{t \in \{0, t_1, \dots, t_M\}}\in \tilde{\Xi} \subseteq X_I \times X^M,
\end{equation}
for $M$ sampled time steps $t_1, \dots, t_M$. 

\begin{property}[Safety]\label{prop:safe}
   Consider \eqref{eq:Dyn}, and let $X_I, X_U \subset X$ with $X_I \cap X_U = \emptyset$ denote an initial and an unsafe set, respectively. 
   We say that $\phi$ encodes a safety property for a trajectory $\xi \in \Xi$ if,
    $$
        \phi(\xi) \defeq \forall t \in [0,T], x(t) \notin X_U.
    $$
\end{property}
By the definition of $\phi$, it follows that verifying that a trajectory exhibits the safety property is equivalent to verifying that a trajectory emanating from the initial set avoids the unsafe set for all time instances, until the horizon $T<\infty$. 


We now define the relevant criteria for a certificate $B$ to verify a safety property.
Assume that $B$ is continuous, so that when considering the supremum/infimum of $B$ over $X$ (already assumed to be bounded) or over some of its subsets, this is well-defined. Consider: 
\begin{align}
	\label{eq:barr_states1}
    &B(x) \leq 0 , \forall x \in X_I,\\
	\label{eq:barr_states2}
    &B(x) > 0, \forall x \in X_U,\\
        &\eval{\frac{dB}{dt}}_{x\in\xi} < \frac{1}{T} \Big (\inf_{x \in X_U}B(x)-\sup_{x \in X_I}B(x) \Big), \label{eq:barr_traj}
\end{align}
where $ \frac{dB}{dt} = \frac{\partial B}{\partial x}f(x)$,
and hence recognise that this depends on the system dynamics $f(x)$.
\begin{assum}[Lipschitz Derivative]\label{ass:Lip}
    Assume that $\frac{dB}{dt}$ is Lipschitz continuous.
\end{assum}
Note that even if $\inf_{x \in X_U}B(x)-\sup_{x \in X_I}B(x) > 0$, i.e., if the last condition encodes an increase of $B$ along the system trajectories, the system still avoids entering the unsafe set. This is established in the proof of Proposition~\ref{cert:barr} below.


Denote by $\psi^s$ the conjunction of \eqref{eq:barr_states1} and \eqref{eq:barr_states2}, and by $\psi^\Delta(\xi)$ the property in \eqref{eq:barr_traj}. 
Notice that the latter depends on $\xi$ as it relates to the derivative along a trajectory. 
We define a barrier certificate as follows.
\begin{defn}[Property Verification \& Certificates]\label{prob:property_cert}
    Given a safety property $\phi(\xi)$, and a function $B\colon \mathbb{R}^n \rightarrow \mathbb{R}$, let $\psi^s$ and $\psi^\Delta (\xi)$ be conditions such that, if$$
            \exists B \colon B \models \psi^s\wedge (\forall \xi \in \Xi) B\models\psi^\Delta(\xi) \implies \phi(\xi), \forall \xi \in \Xi,$$
    then the property $\phi$ is verified for all $\xi \in \Xi$. 
    We then say that such a function $B$ is a \emph{barrier certificate}. 
\end{defn}
   
In words, the implication of Definition \ref{prob:property_cert} is that if a barrier certificate $B$ satisfies the conditions in $\psi^s$, as well as the conditions in $\psi^\Delta (\xi)$, for all $\xi \in \Xi$, then the safety property $\phi(\xi)$ is satisfied for all trajectories $\xi \in \Xi$. 

\begin{prop}[Safety/Barrier Certificate]
\label{cert:barr}
    A function $B \colon \mathbb{R}^n \rightarrow \mathbb{R}$ is a safety/barrier certificate if
	 \begin{equation}
	     B \models \psi^s \wedge (\forall \xi \in \Xi) B\models\psi^\Delta(\xi).
	 \end{equation}
\end{prop}

\begin{proof}
It suffices to show that satisfaction of \eqref{eq:barr_traj} implies safety.
     Integrating \eqref{eq:barr_traj} up to $t\leq T$, we obtain 
     \begin{align}
         &B(x(t)) < B(x(0)) + \frac{t}{T} \Big( \inf_{x \in X_U}B(x)-\sup_{x \in X_I}B(x) \Big ) \nonumber \\
         &\leq \frac{T-t}{T} \sup_{x \in X_I}B(x) + \frac{t}{T} \inf_{x \in X_U}B(x) \nonumber \\
         &\leq \frac{t}{T} \inf_{x \in X_U}B(x) \leq \inf_{x \in X_U}B(x). \label{eq:proof_safety}
     \end{align}
     where the second inequality is since $B(x(0)) \leq \sup_{x \in X_I}B(x)$, as $x(0) \in X_I$. The third inequality is since $\sup_{x \in X_I}B(x) \leq 0$ due to \eqref{eq:barr_states1}, and the last one is since $t\leq T$.
     We thus have
           \begin{align}
       B(x(t)) &< \inf_{x \in X_U}B(x),~ t\in[0,T],
       \end{align}
       i.e. the maximum value along a trajectory is less than the infimum over the unsafe region and hence $x(t) \notin X_U, t=[0,T]$ (notice that $x(0) \notin X_U$ holds since $X_I \cap X_U = \emptyset$). The latter implies that all trajectories that start in $X_I$ avoid entering the unsafe set $X_U$, thus concluding the proof. 
       \end{proof}

To synthesise a certificate, we require complete knowledge of the behaviour $f$ of the dynamical system, to allow us to evaluate the derivative $\frac{dB}{dt}$.  
This may be impractical, and we therefore use data-driven techniques to learn a certificate.
\section{Data-Driven Certificate Synthesis}
\label{sec:learn_certs}

For our analysis, we will treat the initial state as random, distributed according to $\mathbb{P}$ (an appropriate probability space is defined; we gloss the technical details here in the interest of space). 
The support of $\mathbb{P}$ will be the set of admissible initial states (i.e. the initial set $X_I$).

We consider $N$ independent and identically distributed (i.i.d.) initial conditions, sampled according to $\mathbb{P}$, namely 
$\{x^i(0)\}_{i=1}^N \sim \mathbb{P}^N$.
Initializing the dynamics from each of these initial states, we unravel a set of continuous-time trajectories $\{\xi^i\}_{i=1}^N$, also referred to as a multi-sample.  
Since there is no stochasticity in the dynamics, we can equivalently say that trajectories (generated from the random initial conditions) are distributed according to the same probabilistic law; hence, with a slight abuse of notation, we write $\xi\sim \mathbb{P}$.
 We impose the following assumption, standard in the scenario approach~\cite{DBLP:journals/jmlr/CampiG23}. 
 Intuitively, this assumption rules out degenerate problem instances, where we could select the same sample multiple times with nonzero probability.
\begin{assum}[Non-concentrated Mass]\label{ass:non-conc_mass}
	Assume that $\mathbb{P}\{\xi \}=0$, for any $\xi \in \Xi$.
\end{assum}

Since we are now dealing with a sample-based problem, we construct probabilistic certificates and hence probabilistic guarantees on the satisfaction of a given property. 

\subsection{The Discrete-Time Case}

Designing safety certificates that enjoy probabilistic guarantees when it comes to new unseen trajectories has been recently considered in~\cite{DBLP:journals/corr/abs-2502-05510}, however, for discrete-time systems. We review the main developments in this direction, extending these to continuous-time ones. 
Consider a mapping $\mathcal{A}$ such that $B_N = \mathcal{A}(\{\xi^i\}_{i=1}^{N})$ as an algorithm that, based on $N$ samples, returns a certificate $B_N$. 
We call a \emph{compression set} of such an algorithm any subset of the input multi-sample that returns the same certificate.
That is, a sample subset $\mathcal{C}_N \subseteq \{\xi^i\}_{i=1}^{N}$ is a compression set if $\mathcal{A}(\mathcal{C}_N) = \mathcal{A}(\{\xi^i\}_{i=1}^{N})$.
In Algorithm~\ref{algo:sub}, we provide a specific synthesis procedure through which $\mathcal{A}$ (and hence the certificate $B_N$) can be constructed. This algorithm is extended from \cite{DBLP:journals/corr/abs-2502-05510} to continuous-time systems and is discussed in the next section.

\begin{thm}[Probabilistic Guarantees~\cite{DBLP:journals/corr/abs-2502-05510}]
\label{thm:Guarantees}
Consider Assumption~\ref{ass:non-conc_mass}, and let $B_N$ and $\mathcal{C}_N$ be the certificate and compression set, respectively, returned by Algorithm~\ref{algo:sub}.
Fix $\beta \in (0,1)$, and for $k<N$, 
let $\varepsilon(k,\beta,N)$ be the (unique) solution to the polynomial equation in the interval $[k/N,1]$
    \begin{align}
               \frac{\beta}{2N} \sum_{m=k}^{N-1}&\frac{\binom{m}{k}}{\binom{N}{k}}(1-\varepsilon)^{m-N} \nonumber \\
               &+\frac{\beta}{6N}\sum_{m=N+1}^{4N}\frac{\binom{m}{k}}{\binom{N}{k}}(1-\varepsilon)^{m-N} = 1,
     \end{align}
    while for $k=N$ let $\varepsilon(N,\beta,N) =1$. We then have that
    \begin{align}
	    \label{eq:cert_bound}
        &\mathbb{P}^N\big\{ \{\tilde{\xi}^i\}_{i=1}^N \in \tilde{\Xi}^N:~  \\
        &\mathbb{P}\{\tilde{\xi} \in \tilde{\Xi}\colon B_N \not\models \psi^s \wedge \psi^\Delta(\tilde{\xi})) \} \leq \varepsilon(C_N,\beta,N)\big\} \nonumber \geq 1-\beta,
    \end{align}
    where $C_N = |\mathcal{C}_N|$ is the cardinality of the compression set.
\end{thm}

\begin{remark}
Theorem~\ref{thm:Guarantees} is not applicable to continuous-time trajectories.
In particular, the inner probability refers to the probability of a new time-discretised trajectory $\tilde{\xi}$ being (un)safe. As a result, 
states along the state trajectory $\xi$ of the continuous-time system under study may violate the safety property (chiefly, between sampled states $x(t_i)$ and $x(t_{i+1})$).
\end{remark}

\subsection{Extension to the Continuous-Time Case}
In this section we show how to extend Theorem~\ref{thm:Guarantees} to offer guarantees on newly sampled \emph{continuous-time trajectories}, even though, to allow for practical applications, the barrier certificate we will construct will still be based only on time-discretised ones. 
To achieve this, we replace \eqref{eq:barr_traj} by a discretised version based on a first-order Euler approximation. 
Moreover, we tighten it to enforce a increase upper bound condition between sample times.
Denote this condition (over discretised trajectories) as $\psi^\Delta_d(\tilde{\xi})$, defined by the inequality
\begin{align}
     \max_{k = 1,\dots,M} \label{eq:decrease_tight}&\frac{B(x(t_k))-B(x(t_{k-1}))}{t_k-t_{k-1}} \\&<\frac{1}{T} \Big (\inf_{x \in X_U}B(x)-\sup_{x \in X_I}B(x) \Big)-d,\nonumber
\end{align}
where $d \in \mathbb{R}$ is a tightening parameter.
   Define by $\mathcal{L}_B$ and $\mathcal{L}_f$ the Lipschitz constants of the certificate derivative $\frac{\partial B_{\theta^\star}}{\partial x}$ and of the dynamics $f(x)$ respectively, by $\mathcal{M}_B, \mathcal{M}_f$ bounds on their norms, i.e. $\sup_x\|\frac{\partial B_{\theta^\star}}{\partial x}\| \leq \mathcal{M}_B$ and $\sup_x\|f(x)\|\leq\mathcal{M}_f $, and for $\overline{t} = \max_{k=1,\dots,M}(t_{k}-t_{k-1})$, define $d$ as
   \begin{equation}
   \label{eq:dgap}
       d= \overline{t}\mathcal{M}_f\left(\mathcal{M}_B\mathcal{L}_f+\mathcal{M}_f\mathcal{L}_B \right).
   \end{equation}

\begin{thm}[Continuous-Time Guarantees]
\label{thm:CT_PAC}
Consider the conditions of Theorem \ref{thm:Guarantees}, Assumption \ref{ass:Lip} and \eqref{eq:dgap}. Then, 
   \begin{align}
       &\mathbb{P}^N\big\{\{\tilde{\xi}^i\}_{i=1}^N \in \tilde{\Xi}^N\colon \\ 
       & \mathbb{P}\{\xi \in \Xi\colon B_N \not\models \psi^s \wedge \psi^\Delta(\xi)) \} \leq \varepsilon(C_N,\beta,N)\big\} \geq 1-\beta. \nonumber
   \end{align}
\end{thm}

The proof of this is achieved by bounding the difference between the safety of the continuous-time trajectories and their time-discretised approximations; see Appendix. 

\section{Certificate Synthesis}
\label{sec:training}

To learn a barrier certificate from samples, we consider a neural network, a well-studied class of function approximators that generalise well to a given task, although any other function approximation mechanism could be chosen instead. 
We consider a network with $N_\mathrm{layers}$ layers, and $N_\mathrm{neurons}$ neurons per layer.
Denote all tunable neural network parameters by a vector $\theta$ $\in \mathbb{R}^{N_\mathrm{layers} N_\mathrm{neurons}}$. 
We then have that our certificate $B_N$ depends on $\theta$. 
For the results of this section, we write $B_\theta$ and drop the dependency on $N$ to ease notation. 

\begin{algorithm*}
\caption{Certificate Synthesis and Compression Set Computation}
\vspace{0.2cm}
\label{algo:sub}
    \begin{algorithmic}[1]
    \Function{}{}$\mathcal{A}(\theta,\mathcal{D})$
    \State Set $k \gets 1$ \Comment{Initialise iteration index}
            \State Set $\mathcal{C}\leftarrow \emptyset$ 
            \Comment{Initialise compression set}
            \State Fix $L_1 < L_0$ with $|L_1-L_0|>\eta$ \Comment{$\eta$ is any fixed tolerance}
            \While{$l^s(\theta)>0$} \Comment{While sample-independent state loss is non-zero}
                \State  $g \gets \nabla_\theta l^s(\theta)$  \Comment{Gradient of loss function} \label{line:state_grad}
                \State $\theta \leftarrow \theta-\alpha g$ \Comment{Step in the direction of sample-independent gradient} \label{line:state_step}
            \EndWhile
            \Statex \vspace{-0.35cm}\hrulefill
	    \While{$|L_k-L_{k-1}|>\eta $} \Comment{Iterate until tolerance is met} \label{line:while}
            \State $\mathcal{M} \gets \{\tilde{\xi} \in \mathcal{D} \colon L(\theta,\tilde{\xi}) \geq \max_{\tilde{\xi} \in \mathcal{C}} L(\theta,\tilde{\xi}) \}$ \Comment{Find maximal samples with loss greater than compression set loss} \label{line:max_D}
            \State  $\overline{g}_\mathcal{M} \gets \{\nabla_\theta L(\theta,\tilde{\xi})\}_{\tilde{\xi} \in \mathcal{M}}$  \Comment{ Subgradients of loss function for $\tilde{\xi} \in  \mathcal{M}$} \label{line:subgrad_D}
            \State $\overline{\xi}_{\mathcal{C}} \in \argmax_{\tilde{\xi} \in \mathcal{C}} L(\theta,\tilde{\xi})$ \Comment{Find a sample with maximum loss from $\mathcal{C}$} \label{line:max_C}
            \State $\overline{g}_{\mathcal{C}}  \gets \nabla_\theta L(\theta,\overline{\xi}_{\mathcal{C}})$ \Comment{Approximate subgradient of loss function for $\tilde{\xi} = \overline{\xi}_{\mathcal{C}}$} \label{line:subgrad_C}
            \Statex \vspace{-0.25cm}\hrulefill
             \If{$\exists \overline{g} \in \overline{g}_\mathcal{M} \colon \langle \overline{g},\overline{g}_{\mathcal{C}} \rangle \leq 0 \wedge \overline{g} \neq 0$ \label{line:inner}} \Comment{If there is a misaligned subgradient (take the maximum if multiple)}
                \State $\theta \leftarrow \theta-\alpha\overline{g}$ \Comment{Step in the direction of misaligned subgradient} \label{line:true_step}
                \State $\mathcal{C} \leftarrow \mathcal{C} \cup \{\overline{\xi}\}$ \Comment{Update compression set with sample corresponding to $\overline{g}$} \label{line:update_C}
            \Else
                \State $\theta \leftarrow \theta-\alpha\overline{g}_{\mathcal{C}}$ \Comment{Step in the direction of approximate subgradient} \label{line:approx_step}
            \EndIf \label{line:endif}
            \Statex \vspace{-0.35cm}\hrulefill
	    \State $L_k \gets \min\left\{ L_{k-1}, \max_{\tilde{\xi} \in \mathcal{D}} L(\theta,\tilde{\xi})\right\}$ \Comment{Update ``running'' loss value} \label{line:update_L}
        \State $k \gets k+1$ \Comment{Update iteration index} \label{line:update_k}
            \EndWhile
            \State \Return $\theta, \mathcal{C}$
            \EndFunction
    \end{algorithmic}
\end{algorithm*}

\subsection{Certificate Synthesis Algorithm}


We seek to minimise a loss function that encodes the barrier certificate conditions, with respect to the neural network parameters. For $\tilde{\xi} \in \tilde{\Xi}$ and parameter vector $\theta$, let \begin{equation}
    L(\theta,\tilde{\xi})=l^\Delta(\theta,\tilde{\xi})+ l^s(\theta),\label{eq:opt_prob}
\end{equation} represent an associated loss function comprised of sample-dependent loss $l^\Delta$, and sample-independent loss $l^s$
\begin{align}
    l^s(\theta) &\defeq \frac{1}{|\mathcal{X}_I|}\sum_{x \in \mathcal{X}_I}\max\{0,B_\theta(x)\}\nonumber \\
    &+\frac{1}{|\mathcal{X}_U|}\sum_{x \in \mathcal{X}_U} \max\{0,\delta-B_\theta(x)\}.\nonumber\\
    l^\Delta(\theta, \tilde{\xi}) &\defeq -\frac{1}{T} \Big (\inf_{x \in \mathcal{X}_U}B_\theta(x)-\sup_{x \in \mathcal{X}_I}B_\theta(x) \Big) \nonumber\\&+\max_{k = 1,\dots,M} \frac{B_\theta(x(t_k))-B_\theta(x(t_{k-1}))}{t_k-t_{k-1}},\nonumber
    \end{align}
To instantiate these functions we consider a discrete set of $N_{\mathrm{states}}$ grid-points on each sub-domain: $\mathcal{X}_{I}$ is the set of points in the initial set, and $\mathcal{X}_U$ is the set of points in the unsafe set.
These points are generated densely enough across the domain of interest, and hence offer an accurate approximation. 
Since they do not require access to the dynamics, we consider them separately from $\{\tilde{\xi}^i\}_{i=1}^N$.

To see the choice of the loss functions, consider the first summation in $l^s(\theta)$, and notice that if $B_\theta(x) \leq 0$ then $\max\{0, B_\theta(x)\}=0$, i.e., no loss is incurred, implying satisfaction of \eqref{eq:barr_states1}. 
Similar reasoning relates the other summation and expression of $ l^\Delta(\theta, \tilde{\xi})$ to \eqref{eq:barr_states2} and \eqref{eq:decrease_tight}, respectively. 

The following mild assumption ensures a minimiser of the loss functions exists and hence our algorithm terminates.
\begin{assum}[Minimisers' Existence] \label{ass:exist}
For any $\{\tilde{\xi}\}_{i=1}^N$, and any nonempty $\mathcal{D} \subseteq \{\tilde{\xi}\}_{i=1}^N$, the set of minimisers of $\max_{\tilde{\xi} \in \mathcal{D}} L(\theta,\tilde{\xi}),$ is nonempty.
\end{assum}

Algorithm \ref{algo:sub} provides an inexact subgradient methodology to minimise the loss function, and to iteratively construct a compression set $\mathcal{C}$ (initially empty; see step $3$). 
We explain the main steps of this algorithm with reference to Figure \ref{fig:algorithm}, where each sample gives rise to a concave triangular constraint. 
After an arbitrary initialization, we follow the subgradient associated with the worst case sample and add it to the compression set $\mathcal{C}$ (step $14$--$15$, point $M_1$). 
When iterates get to point $M_2$, the subgradient step becomes inexact, as for the same parameter there exists a sample resulting in a higher loss (see asterisk). 
Such a sample is in $\mathcal{M}$, step $9$ of Algorithm \ref{algo:sub}. 
However, the algorithm does not ``jump'', as the condition in step $13$ of the algorithm is not yet satisfied. 
Graphically, this is since the $M_2$ and the red asterisk are on a side of the respective constraint with the same slope. 
As such the algorithm performs inexact subgradient descent steps up to point $M_3$; this is the first instance where the condition in step $13$ is satisfied (there exists another constraint with opposite slope) and hence the algorithm ``jumps'' to a point with higher loss and subgradient of opposite sign. 
This procedure is then repeated as shown in the figure, with the red line indicating the iterates' path. The ``jumps'' serve as an exploration step to investigate the non-convex landscape, while their number corresponds to the cardinality of the returned compression set. This can be thought of as a constructive procedure for the general framework presented in \cite{DBLP:conf/nips/PaccagnanCG23} to construct compression sets.

The computational complexity of one iteration of the main loop is $\mathcal{O}((M+N_{\mathrm{states}})\cdot N \cdot N_\mathrm{layers}\cdot N_\mathrm{neurons}^3)$.
We can upper bound the number of iterations by $\frac{L_0-L^\star}{\eta}$, where $L^\star$ is the value of the loss associated with one of the minimisers. To terminate our algorithm we require knowledge of $\dgap$ (which depends on $\theta^\star$).
To resolve this, we propose two approaches.
\begin{enumerate}
	\item At every iteration $j$ calculate $\dgap_j$ using the current best parameters $\theta_j$, terminate if $\max_i\left[L(\theta_j, \tilde{\xi}^i)\right] < -\dgap_j$. 
\item Choose a parameter set, and consider the supremum across that set to determine an upper bound on $\mathcal{L}_B, \mathcal{M}_B$. Use these upper bounds to calculate $\dgap$.
\end{enumerate}
The second option is likely to be conservative but removes the need for calculating Lipschitz constants in the loop, hence is computationally more efficient.
To determine Lipschitz constants for neural networks we refer to \cite{DBLP:conf/nips/FazlyabRHMP19,DBLP:conf/nips/VirmauxS18}.


\begin{figure}[h!]
\centering
\begin{tikzpicture}[>=latex, font=\large, scale = 1.15]   
\draw[line width=1.25pt] plot[domain=-4.2:-1] (\x,{-1*(\x + 4.2)+3.2});      
\draw[line width=1.25pt] plot[domain=-5.3:-4.2] (\x,{1*(\x + 4.2)+3.2});   
\fill[pattern=vertical lines] (-4.2,3.2) -- (-1,0) -- (-1,-0.3) -- (-4.2,2.9) -- cycle;      
\fill[pattern=vertical lines] (-5.3,2.1) -- (-4.2,3.2) -- (-4.2,2.9) -- (-5.3,1.8) -- cycle;

\draw[line width=1.25pt] plot[domain=-3.4:-0.93] (\x,{-1.6*(\x + 3.4)+3.96});      
\draw[line width=1.25pt] plot[domain=-5.3:-3.4] (\x,{1.6*(\x + 3.4)+3.96});   
\fill[pattern=vertical lines] (-3.4,3.96) -- (-0.93,0) -- (-0.93,-0.3) -- (-3.4,3.66) -- cycle;      
\fill[pattern=vertical lines] (-5.3,0.92) -- (-3.4,3.96) -- (-3.4,3.66) -- (-5.3,0.62) -- cycle;

\draw[line width=1.25pt] plot[domain=-0.5:0.9] (\x,{-0.85*(\x + 0.5)+3});      
\draw[line width=1.25pt] plot[domain=-4.02:-0.5] (\x,{0.85*(\x + 0.5)+3});   
\fill[pattern=vertical lines] (-0.5,3) -- (0.9,1.81) -- (0.9,1.51) -- (-0.5,2.7) -- cycle;      
\fill[pattern=vertical lines] (-4.02,0) -- (-0.5,3) -- (-0.5,2.7) -- (-4.02,-0.3) -- cycle;

\draw[line width=1.25pt] plot[domain=0.2:0.9] (\x,{-2.2*(\x - 0.2)+3.7});      
\draw[line width=1.25pt] plot[domain=-1.48:0.2] (\x,{2.2*(\x - 0.2)+3.7});   
\fill[pattern=vertical lines] (0.2,3.7) -- (0.9,2.16) -- (0.9,1.86) -- (0.2,3.4) -- cycle;      
\fill[pattern=vertical lines] (-1.48,0) -- (0.2,3.7) -- (0.2,3.4) -- (-1.48,-0.3) -- cycle;

%
\draw[->,thick] (-5.5,0)--(1.5,0) node[below]{};
\draw[->,thick] (-5.3,-0.2)--(-5.3,4.5) node[left]{Loss};
\draw[->,thick, red!80] (-4.1,3.25)-- (-3.25,2.4) node[below]{};
\draw[->,thick, red!80] (-3.25,2.4) -- (-2.4,1.55) node[below]{};
\draw[->,thick, red!80] (-2.4,1.55) -- (-1.55,0.7) node[below]{};
\draw[->,thick, red!80] (-1.55,0.7) --  (-1.55,2.25) node[below]{};
\draw[->,thick, red!80] (-1.55,2.25) -- (-2.2,1.7) node[below]{};
\draw[->,thick, red!80] (-2.2,1.7) -- (-2.2,2.25) node[below]{};
\draw[->,thick, red!80] (-2.2,2.25) -- (-1.95,1.85) node[below]{};
\filldraw[black] (-4.1,3.25) circle (2pt);
\filldraw[black] (-3.25,2.4) circle (2pt);
\filldraw[black] (-2.4,1.55) circle (2pt);
\filldraw[black] (-1.55,0.7) circle (2pt);
\filldraw[black] (-1.55,2.25) circle (2pt);
\filldraw[black] (-2.2,1.7) circle (2pt);
\filldraw[black] (-2.2,2.25) circle (2pt);
\draw[] (0.2,-0.2)node[below, bag]{Parameter Space};
\draw[]  (-4.1,3.3)node[above, bag]{$M_1$};
\draw[]  (-3.25,2.5)node[above, bag]{$M_2$};
\draw[red!80] (-3.25,3.8) node[above, bag]{$\large{\boldsymbol{\ast}}$};
\draw [->,ultra thick] (0,1.2) to [bend right=30] (-1.5,0.8);
\draw[]  (0.25,0.9) node[above, bag]{$M_3$};
\end{tikzpicture}
\caption{Graphical illustration of Algorithm~\ref{algo:sub}.}
\label{fig:algorithm}
\end{figure}



\subsection{Sample Discarding and Compression Computation}
Due to the tightening introduced, to minimise the loss function, we would like its value to be lower than or equal to $-d$. However, in some cases, the parameter returned by Algorithm \ref{algo:sub} may result in a loss value greater than $-d$, thus preventing us from verifying safety. To ensure the loss is below that threshold, we employ a sampling-and-discarding procedure \cite{DBLP:journals/jota/CampiG11,DBLP:journals/tac/RomaoPM23}, and introduce Algorithm \ref{algo:main} as an outer loop around Algorithm \ref{algo:sub}. 
Algorithm \ref{algo:main} iterates as long the desired loss value is not met (step $4$), and progressively discards samples (step $7$). 
The samples that are discarded are the ones identified as compression by Algorithm \ref{algo:sub} (step $6$), as removing these is bound to improve the solution. 

\begin{algorithm}[ht]
\caption{Compression Set Update with Discarding}
\vspace{0.2cm}
\label{algo:main}
\begin{algorithmic}[1]
\State Fix $ \{\tilde{\xi}^i\}_{i=1}^N$
    \State Set $\widetilde{\mathcal{C}}\gets \emptyset$\Comment{Initialise compression set}
    \State Set $\mathcal{D} \gets \{\tilde{\xi}^i\}_{i=1}^N$ \Comment{Initialise ``running'' samples}
    \While{$(\max_{\tilde{\xi} \in \mathcal{D}} l^\Delta(\theta, \tilde{\xi})>-\dgap) \bigvee (l^s(\theta) > 0) $}
            \State $\theta, \mathcal{C} \gets$ $\mathcal{A}(\theta,\mathcal{D})$ \Comment{Call Algorithm \ref{algo:sub}} \label{line:subgrad}
        \State $\widetilde{\mathcal{C}} \gets \widetilde{\mathcal{C}} \cup \mathcal{C}$ \Comment{{Update $\widetilde{\mathcal{C}}$}} \label{line:update_outer_C}
        \State  $\mathcal{D} \gets \mathcal{D} \setminus \widetilde{\mathcal{C}}$ \Comment{Discard $\widetilde{\mathcal{C}}$ from $\mathcal{D}$} \label{line:discard}
    \EndWhile
        \State \Return $\theta$, $\widetilde{\mathcal{C}}$
\end{algorithmic}
\end{algorithm}

The samples that are discarded need to be added to the compression set (step $6$); intuitively this is the case as if we repeat the procedure using only the compression samples we need to follow the same solution path. A direct byproduct of Algorithms~\ref{algo:sub} and~\ref{algo:main} is that they calculate a compression set, whose cardinality $C_N$ is in turn used in Theorem \ref{thm:CT_PAC} to provide the desired probabilistic guarantees. The cardinality of that set corresponds to the number of ``jumps'' in Algorithm~\ref{algo:sub} plus the number of discarded samples in Algorithm~\ref{algo:main}. 
This cardinality is not necessarily greater for higher dimensional problems, but is rather case dependent and depends on the complexity of the problem.  
For example, a problem where some trajectories approach or even enter the unsafe set presents a more challenging synthesis problem than one where trajectories all move in the opposite direction to the unsafe set, thus we expect the former to have a larger compression set even if the problem is smaller in dimension. 
\section{Numerical Results}
\label{sec:exp}
We consider constructing a safety certificate for the nonlinear, two-dimensional jet engine model as considered in~\cite{DBLP:journals/tac/NejatiLJSZ23},
\hypersetup{hidelinks}\blfootnote{The codebase is available at \url{https://github.com/lukearcus/fossil_scenario}}\hypersetup{pdfborder={0 0 1}}
\begin{equation}
\label{eq:jet}
        \dot{x}_1(t) = -x_2(t)-\frac{3}{2}x_1^2(t)-\frac{1}{2}x_1^3(t), \;\;
        \dot{x}_2(t) = x_1(t),
\end{equation}
using a time horizon $T=5$.
Figure~\ref{fig:jet_eng_comparison} provides a graphical representation of the dynamics, subdomains under study, the $0$-level set produced by our certificate, and the level sets calculated by the methods in \cite{DBLP:journals/tac/NejatiLJSZ23} (one lower bounding the unsafe set, the other upper bounding the initial set). 
We used $5$ independent repetitions (each with different multi-samples) of $1,000$ sampled trajectories and $367$ seconds of computation time (standard deviation $139$s), to obtain $\varepsilon = 0.01492$ (standard deviation $0.00140$) with confidence $0.99$.
The methodology of \cite{DBLP:journals/tac/NejatiLJSZ23} required $257149$ state pair samples and $5123$ seconds (standard deviation $449$s) of computation time to compute a barrier certificate with the same confidence (however, this holds deterministically). 
We estimate the Lipschitz constants using the methodology in \cite{DBLP:journals/jgo/WoodZ96}, noting that convergence is guaranteed asymptotically.
Figure~\ref{fig:jet_barr_surf} contains a 3D plot of the certificate.

 \begin{figure}[t]
    \begin{minipage}[t]{0.48\linewidth}
        	\includegraphics[width=\linewidth]{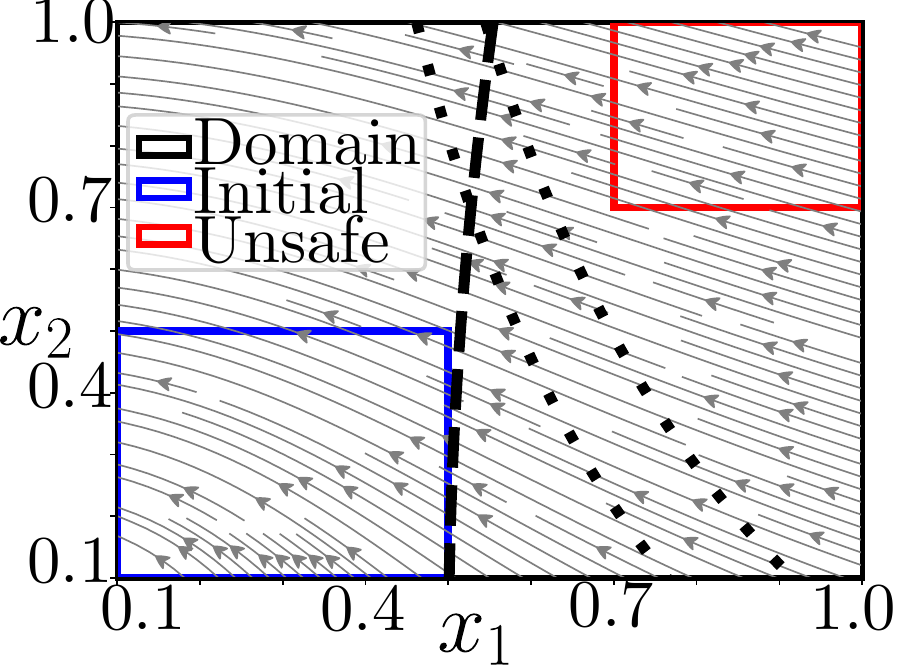}
	\caption{Phase plane plot, initial and unsafe set for \eqref{eq:jet}. The zero-level set for our certificate is dashed; level sets that bound the initial and unsafe sets in~\cite{DBLP:journals/tac/NejatiLJSZ23} are dotted.}
	\label{fig:jet_eng_comparison}
    \end{minipage}\hfill
     	\begin{minipage}[t]{0.48\linewidth}
	   \includegraphics[width=\linewidth]{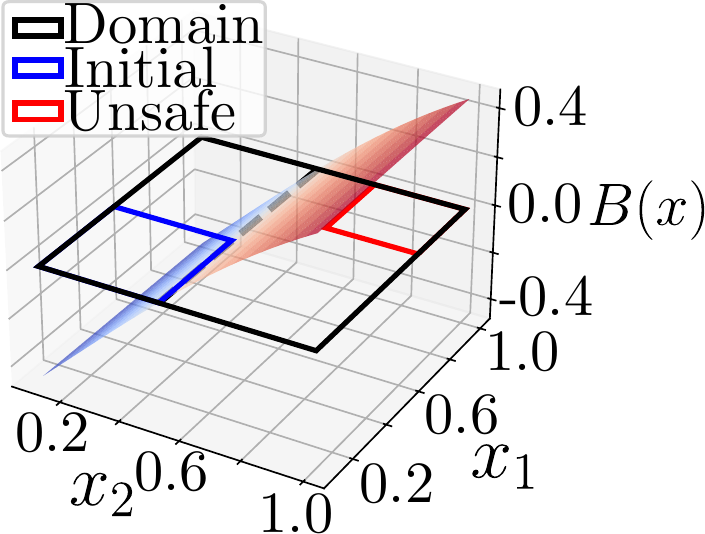}
	   \caption{Surface plot of the safety/barrier certificate, generated by our techniques, for the system of Figure \ref{fig:jet_eng_comparison}.}
	   \label{fig:jet_barr_surf}
    \end{minipage}
 \end{figure}

Beyond these numerical results, we briefly discuss the theoretical differences between our approach and \cite{DBLP:journals/tac/NejatiLJSZ23}. 
The results in \cite{DBLP:journals/tac/NejatiLJSZ23} offer a guarantee that, with a certain confidence, the safety property is \emph{always} satisfied, in contrast to Theorem \ref{thm:CT_PAC} where we provide such guarantees in probability (up to a quantifiable risk level $\varepsilon$). 
 However, these ``always'' guarantees, although very useful, come with some challenges. 
 Firstly, they are not applicable when part of the initial set is unsafe, whereas we can still bound the probability of a new trajectory being safe even if some of the sampled trajectories were unsafe.
Second, they implicitly require knowledge of the underlying probability distribution to instantiate their confidence bound and scale exponentially with respect to the state space dimension. 

Related to this last point, we 
 performed a comparison on the following four-dimensional system taken from~\cite{DBLP:conf/hybrid/EdwardsPA24}. 
\begin{equation}
\begin{aligned}
        \dot{x}_1(t) &= x_1(t) + \frac{x_1(t)  x_2(t)}{5} -\frac{x_3(t)x_4(t)}{2},\\
        \dot{x}_2(t) &= \cos(x_4(t)),\\
        \dot{x}_3(t) &= 0.01\sqrt{|x_1(t)|},\\
        \dot{x}_4(t) &= -x_1(t) - x_2(t)^2 + \sin(x_4(t)),
\end{aligned}
\end{equation}
using $T=4$.
We applied also the approach of \cite{DBLP:journals/tac/NejatiLJSZ23} which, with $10^{19}$ samples, returned a confidence $10^{-30}$, an uninformative result as it is close to zero. 
With only $100$ samples our techniques obtain a risk level $\varepsilon = 0.21450$ (standard deviation $0.00910$), confidence $1-10^{-5}$ (close to one). 

\section{Conclusion}

We have proposed a method for synthesis of neural-network certificates for continuous-time dynamical systems, based only on a finite number of trajectories from a system.   
Our numerical experiments demonstrate the efficacy of our methods on a number of examples, involving comparison with related methodologies in the literature. 
Current work concentrates towards extending to controlled systems, thus co-designing a controller and a certificate at the same time.




  \section*{APPENDIX}
 \setcounter{section}{0}
 \section{Proof of Theorem \ref{thm:CT_PAC}}

We aim at finding a bound on the discretisation gap $L(\theta,\xi)-L(\theta,\tilde{\xi})$ so that, for sufficiently small loss evaluated on the time-discretised approximations $L(\theta,\tilde{\xi})$, we also achieve a negative loss on the continuous trajectories $L(\theta,\xi)$.

\begin{align}
    &L(\theta,\xi)-L(\theta,\tilde{\xi})
    =l^\Delta(\theta,\xi)- l^\Delta(\theta,\tilde{\xi}) \nonumber \\
    &=\max_{x \in \xi} \eval{\frac{dB}{dt}}_{x}-\max_{k = 1,\dots,M} \frac{B(x(t_k))-B(x(t_{k-1}))}{t_k-t_{k-1}}. \label{eq:proof1}
    \end{align}
    Replace the first maximisation with one between time instances, and exchange the order of the $\max$ operators,
    \begin{align}
    &\max_{k=1,\dots,M}\max_{t\in[t_{k-1},t_k]}\eval{\frac{dB}{dt}}_{x(t)}\\&\qquad-\max_{k=1,\dots,M}\frac{B(x(t_k))-B(x(t_{k-1}))}{t_k-t_{k-1}},\nonumber\\
    &\leq\max_{k=1,\dots,M}\Bigg[\max_{t\in[t_{k-1},t_k]}\eval{\frac{dB}{dt}}_{x(t)}-\frac{B(x(t_k))-B(x(t_{k-1}))}{t_k-t_{k-1}}\Bigg].\label{eq:proof2}
    \end{align}
    We can now replace the difference term with an integral, 
    \begin{align}
            &\max_{k=1,\dots,M}\frac{\int_{t_{k-1}}^{t_{k}}\max_{t\in[t_{k-1},t_k]}\eval{\frac{dB}{dt}}_{x(t)}-\eval{\frac{dB}{dt}}_{x(\tau)}~\mathrm{d}\tau}{t_k-t_{k-1}}.\nonumber
    \end{align}
    Letting $\mathfrak{L} = \mathcal{M}_B\mathcal{L}_f+\mathcal{M}_f\mathcal{L}_B$ (refer to \eqref{eq:dgap} for the definition of the various constants), the previous derivations lead to
    \begin{align}
        &L(\theta,\xi)-L(\theta,\tilde{\xi}) \nonumber\\
        &\leq\max_{k=1,\dots,M}\frac{\int^{t_k}_{t_{k-1}}\|x(\tau)-\max_{t\in[t_{k-1},t_k]}x\|\mathfrak{L}~\mathrm{d}\tau}{t_k-t_{k-1}}\nonumber \\
        &\leq\max_{k=1,\dots,M}\mathfrak{L}\frac{\int_{t_k}^{t_{k-1}} \mathcal{M}_f (t_k-t_{k-1}) ~\mathrm{d}\tau}{t_k-t_{k-1}},\\
        &=\max_{k=1,\dots,M}\mathfrak{L}\int_{t_k}^{t_{k-1}} \mathcal{M}_f ~\mathrm{d}\tau=\overline{t}\mathfrak{L}\mathcal{M}_f.
    \end{align}
    where the second inequality is since $\sup_x\|f(x)\| \leq \mathcal{M}_f$, and the last one since $\overline{t} = \max_{k=1,\dots,M}(t_{k}-t_{k-1})$.

This results then to a discretisation gap as in \eqref{eq:dgap}.
By Theorem~\ref{thm:Guarantees}, and noticing that violating the conditions with $\psi^{\Delta}_d$ in place of $\psi^{\Delta}$, is equivalent to $L(\theta^\star, \tilde{\xi}) >-d$, we have
  \begin{equation*}
   \begin{aligned}
       \mathbb{P}^N\left\{\{\tilde{\xi}^i\}_{i=1}^N\colon\mathbb{P}\{\tilde{\xi} \colon L(\theta^\star, \tilde{\xi}) >-d \} \leq \varepsilon(C_N,\beta,N)\right\}
       \geq 1-\beta.
   \end{aligned}
   \end{equation*}
   Since $L(\theta^\star, \xi) \leq L(\theta^\star, \tilde{\xi})+d$, this then implies that
     \begin{equation*}
   \begin{aligned}
       \mathbb{P}^N\left\{\{\tilde{\xi}^i\}_{i=1}^N\colon\mathbb{P}\{\xi \colon L(\theta^\star, \xi) > 0 \} \leq \varepsilon(C_N,\beta,N)\right\}
       \geq 1-\beta,
   \end{aligned}
   \end{equation*}
   thus concluding the proof. \hfill \qed


\bibliographystyle{ieeetr}
\bibliography{Bibliography}


\end{document}